\documentclass[aps,superscriptaddress,nofootinbib]{revtex4-2}
\setcitestyle{numbers,open={(},close={)}}

\usepackage[letterpaper, left=0.8in, right=0.8in, top=0.7in, bottom=0.7in]{geometry}

\usepackage{amssymb}
\usepackage{amsmath,mathtools}
\usepackage{amsthm}
\usepackage{braket}
\usepackage{float}

\usepackage{lmodern}
\usepackage[T1]{fontenc}

\usepackage{graphicx}%
\graphicspath{{./figures/}}

\usepackage{hyperref}
\hypersetup{
	colorlinks=true,
	linkcolor=blue,
	filecolor=blue,
	citecolor=black,      
	urlcolor=cyan,
}

\usepackage{dcolumn}%
\usepackage{bm}%

\usepackage{lipsum}
\usepackage{ifdraft}
\usepackage{color}

\newcommand{\figref}[1]{Fig.~\ref{#1}}
\newcommand{\appref}[1]{Appendix~\ref{#1}}

\newcommand{\eqnref}[1]{Eq.~(\ref{#1})}

\newcommand{\renyientropy}{R\'enyi entropy}

\DeclareMathOperator{\Tr}{Tr}

\newcommand{\TEE}[1]{I^{(#1)}_3}
\newcommand{\Dtot}[1]{D_{\mathrm{tot}}}
\newcommand{\tcr}{\mathcal{T}_{\mathrm{CR}}}
\newcommand{\tcl}{\mathcal{T}_{\mathrm{CL}}}

\usepackage{nicefrac}

\usepackage{orcidlink}

\newcommand{\pheader}[1]{\noindent \textbf{#1 } \newline}

\begin{document}

\title{%
Quantum computation at the edge of chaos}

\author{Tomohiro Hashizume\,\orcidlink{0000-0002-7154-5417}}
\email[]{tomohiro.hashizume@uni-hamburg.de}
\affiliation{
   Institute for Quantum Physics, University of Hamburg, Luruper Chaussee 149, 22761 Hamburg, Germany
}
\affiliation{The Hamburg Centre for Ultrafast Imaging, Luruper Chaussee 149, 22761, Hamburg, Germany}
\affiliation{Max Planck Institute for the Structure and Dynamics of Matter, Luruper Chaussee 149, 22761 Hamburg, Germany}

\author{Zhengjun Wang\,\orcidlink{0000-0002-7599-0382}}
\affiliation{
   Institute for Quantum Physics, University of Hamburg, Luruper Chaussee 149, 22761 Hamburg, Germany
}
\affiliation{The Hamburg Centre for Ultrafast Imaging, Luruper Chaussee 149, 22761, Hamburg, Germany}
\affiliation{Max Planck Institute for the Structure and Dynamics of Matter, Luruper Chaussee 149, 22761 Hamburg, Germany}

\author{Frank Schlawin\,\orcidlink{0000-0001-6977-5881}}
\affiliation{
   Institute for Quantum Physics, University of Hamburg, Luruper Chaussee 149, 22761 Hamburg, Germany
}
\affiliation{The Hamburg Centre for Ultrafast Imaging, Luruper Chaussee 149, 22761, Hamburg, Germany}
\affiliation{Max Planck Institute for the Structure and Dynamics of Matter, Luruper Chaussee 149, 22761 Hamburg, Germany}

\author{Dieter Jaksch\,\orcidlink{0000-0002-9704-3941}}
\affiliation{
   Institute for Quantum Physics, University of Hamburg, Luruper Chaussee 149, 22761 Hamburg, Germany
}
\affiliation{The Hamburg Centre for Ultrafast Imaging, Luruper Chaussee 149, 22761, Hamburg, Germany}
\affiliation{
   Clarendon Laboratory, University of Oxford, Parks Road, Oxford OX1 3PU, United Kingdom
}

\date{\today}

\begin{abstract}

A key challenge in classical machine learning is to mitigate
overparameterization by selecting \textit{sparse} solutions. 
We translate this concept to the quantum domain, introducing \textit{quantum sparsity} as a principle based on 
minimizing quantum information shared across multiple parties. 
This allows us to address fundamental issues in quantum data processing and 
convergence issues such as the barren plateau problem in Variational Quantum Algorithm (VQA). 
We propose a practical implementation of this principle using the topological Entanglement Entropy (TEE) as a cost function 
regularizer. 
A non-negative TEE is associated with states with a sparse structure in a suitable basis, 
while a negative TEE signals untrainable chaos.
The regularizer, therefore, guides the optimization along the critical \textit{edge of chaos} that separates these regimes. 
We link the TEE to structural complexity by analyzing quantum states encoding functions of tunable smoothness, 
deriving a quantum Nyquist-Shannon sampling theorem that bounds the resource requirements and error propagation in VQA. 
Numerically, our TEE regularizer demonstrates significantly improved convergence and precision for complex data encoding 
and ground-state search tasks. 
This work establishes quantum sparsity as a design principle for robust and efficient VQAs.

\end{abstract}

\maketitle
Quantum computing has the potential to revolutionize information processing and speed up fundamental problems in materials science, machine learning, and data analysis. 
This potential, however, hinges on our ability to design quantum circuits that avoid uncontrolled scrambling 
of quantum information that would otherwise lead to chaotic dynamics. 
A prominent example where this problem becomes apparent are the variational quantum algorithms (VQAs)
\cite{cerezoVariationalQuantumAlgorithms2021,jakschVariationalQuantumAlgorithms2023}. 
Those hybrid quantum-classical algorithms are known to suffer from a fundamental problem that threatens 
to limit or even completely destroy its practical use: barren plateaus, 
where the gradients of the loss function required for optimization of parameters $\bm{\theta}$
vanish exponentially with system size, leading to the failure of classical optimization
\cite{mccleanBarrenPlateausQuantum2018,laroccaBarrenPlateausVariational2025}.
Such gradient concentration is intrinsically linked to the circuit’s transition into a chaotic regime; 
here, the high expressibility of the ansatz, necessitated by the circuit depth $D$, 
facilitates rapid information scrambling \cite{anthonychenSpeedLimitsLocality2023}. 
As illustrated in \figref{fig:overview}a, an output state $\ket{\psi_D(\bm{\theta})}$ at depth $D$,
generated via local brickwork layers of parameterized gates (red lines), 
eventually enters a regime of high expressibility, 
where, however, the solution landscape becomes chaotic and untrainable due to rapid information scrambling. 
For VQAs to be successful, their variational quantum circuits must be sufficiently complex to encompass
good solution ansatzes without the circuit being driven beyond this edge of chaos into a chaotic, untrainable parameter regime.
It is this seemingly inevitable trade-off that threatens to destroy the potential advantage of near-term quantum computation and severely restrict the utility of quantum computing for data science more broadly.

A similar competition between complexity and untrainable chaos is also observed in classical computing, 
where it is known as the edge of chaos hypothesis 
\cite{langtonComputationEdgeChaos1990,wolframNewKindScience2002,zhangIntelligenceEdgeChaos2024}.
It posits that the capacity for complex computation emerges at the boundary between simple, ordered dynamics
and chaotic dynamics. 
At this edge, a system is sufficiently stable to store information,
yet flexible enough to process it nontrivially.
This regime avoids both trivial simplicity and the featureless chaos also encountered in barren plateaus, and offers the maximal computational capacity.

\begin{figure*}[t]
   \includegraphics[scale=1.0]{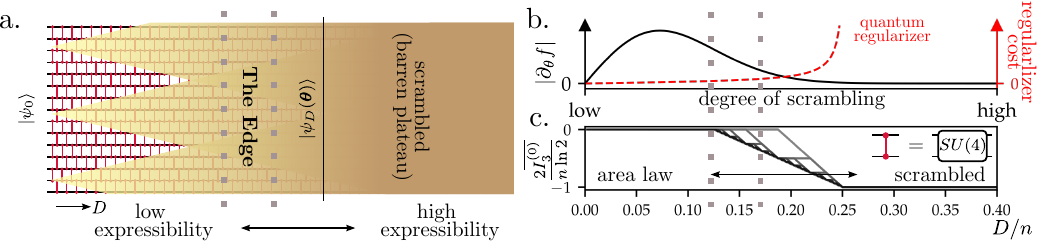}
   \caption{
      \textbf{Information spreading and scrambling in a brickwork circuit.}
      {\bf a.}
      The onset of scrambling in the brickwork circuit and the emergence of a barren plateau.
      The circuit diagram of a brickwork circuit, where qubits (black lines) are evolved from the initial state $\ket{\psi_0}$,
      via entangling parameterized nearest-neighbor gates (red).
      In a shallow circuit (small $D$), 
      the lightcones (shaded yellow) produced by the local perturbations, introduced by changes in the parameters at
      shallow depth, remain local. 
      However, when this lightcone covers the whole circuit (large $D$), local information is scrambled across the system.
      Furthermore, the overlap of these lightcones (shaded brown) exponentially dilutes the local effects on the output state 
      at depth $D$, $\ket{\psi_D(\bm{\theta})}$,
      rendering the circuit indistinguishable from a random unitary circuit.
      This diffusion of structured information hinders gradient-based learning at greater depth with higher expressibility, 
      leading to a barren plateau. The regime that ensures both trainability and high expressibility 
      lies within the critical regime at \textit{the edge}.
      {\bf b.}~Illustrated is the relation between the gradient of a function $f$,
      evaluated using a parameterized quantum circuit, with respect to a single circuit parameter $\theta\in\bm{\theta}$. 
      In conventional learning, as the training progresses, the degree of information scrambling increases, 
      reaching a barren plateau. 
      A quantum regularizer (red) that prevents scrambling, 
      ensuring that training occurs at the quantum edge region (between the gray dotted lines) 
      of the landscape with a nonvanishing gradient. 
      {\bf c.}~The change of the TEE $\TEE{0}$ with respect to the circuit depth $D$. 
      The TEE $\TEE{0}$ becomes increasingly negative with the circuit depth $D$, 
      depicting a clear transition through the edge (between the gray dotted lines) to chaos 
      (number of qubits $n=2^{3},2^{4},\ldots,2^{9}$, light to dark).
      \label{fig:overview}} 
\end{figure*}

To overcome the barren plateau problem, we may therefore look at classical optimization for inspiration. 
There, this problem is addressed using regularizers that penalize overparameterization,  %
such as the $\ell_1$ or $\ell_2$ norm of the parameters~\cite{tibshiraniRegressionShrinkageSelection1996,
donohoMostLargeUnderdetermined2006}, 
thus selecting a sparse solution~\cite{candesStableSignalRecovery2006,donohoCompressedSensing2006}. 
This line of argument suggests that 
quantum computing could benefit from establishing a quantum version of \textit{sparsity}.
It offers a new perspective on the emergence of barren plateaus and, most importantly, allows the development of quantum regularizers to protect the algorithm's convergence (\figref{fig:overview}b). 
It is easy to see that developing such a concept is nontrivial: 
penalizing the $\ell_p$-norm of $\bm{\theta}$ in a parameterized quantum circuit would effectively restrict the circuit depth, thereby limiting dynamics to a simple subclass of states such as weakly entangled matrix product states, 
and potentially forgoing any quantum advantage.
Furthermore, the sparsity of $\bm{\theta}$ typically does not immediately translate to the sparsity of the output state in the computational basis.

In this article, we propose and analyze a highly intuitive notion of quantum sparsity 
as a measure of delocalized quantum information, 
and demonstrate that the topological entanglement entropy (TEE) can be used as a sensitive measure thereof. 
We first show the properties of TEE and its relation to the underlying sparsity by numerically investigating the TEE 
for quantum states encoding functions with tunable differentiability, continuity, and fractal dimension, formally linking a state's structural complexity to quantum chaotic signatures in the encoding quantum circuit. 
This leads us to devise a quantum Nyquist-Shannon sampling theorem (QNSST) for amplitude-encoded data. 
It establishes a polylogarithmic bound on the required quantum resources for efficiently and faithfully amplitude encoding complex classical data. 
Finally, we propose and benchmark the performance of TEE as a quantum regularizer to avoid the barren plateau problem. 
We show up to two orders of magnitude improved convergence in two distinct problems: 
(i) the high-fidelity preparation of a quantum state representing turbulent fluid flow, and 
(ii) a Variational Quantum Eigensolver (VQE) task involving incommensurate circuit and problem geometries.
In both cases, we demonstrate that the TEE regularizer successfully avoids barren plateaus by constraining the optimization to a trainable parameter regime.

In an overdetermined optimization problem, the loss function is minimized over a multidimensional subspace. 
We define the quantum sparsity of a solution, for a pure $n$-qubit state, 
as the amount of nonlocally stored (quantum) information across multiple qubits. 
The sparsest solution is then selected by minimizing a measure of this delocalized quantum information, 
such as the TEE
\cite{kitaevTopologicalEntanglementEntropy2006,hashizumeTunableGeometriesSparse2022}. 
For three subregions $A$, $B$, and $C$, the TEE is defined as
\begin{align}
   I^{(\alpha)}_3(A,B,C) &= I^{(\alpha)}_2(A,B) + I^{(\alpha)}_2(A,C) - I^{(\alpha)}_2(A,BC), 
\end{align}
in terms of the mutual information $I^{(\alpha)}(A,B)= S^{(\alpha)}(A) + S^{(\alpha)}(B) - S^{(\alpha)}(AB)$,
formulated using the \renyientropy{} of order $\alpha$. 
This entropy is defined by $S^{(\alpha)}(A) = \frac{1}{1-\alpha} \ln \Tr \{ \rho_{A}^\alpha \}$, in which the reduced
density operator (RDO), $\rho_A$, is obtained by taking the partial trace of the system's density operator $\rho$ over 
the complementary region $\bar{A}$, such that $\rho_A = \mathrm{Tr}_{\bar{A}} \rho  $.

In the context of variational quantum circuits, 
a large negative TEE naturally measures the degree to which information regarding a local change is encoded globally; 
effectively, it quantifies how the circuit scrambles local changes in parameters into an untrainable, nonlocal form. 
For brickwork circuits, the amount of entanglement generated scales at most linearly with $D$ 
due to the Lieb-Robinson bound \cite{anthonychenSpeedLimitsLocality2023}. 
Consequently, as the circuit deepens, the change in a parameterized local gate 
is diluted across the qubits within its lightcone, as depicted in \figref{fig:overview}a. 
Conversely, the RDO of a subregion of an output state is affected by all gates that are causally connected, 
leading to overparameterization for sufficiently deep circuits.
TEE thus provides a diagnostic for the onset of this regime, 
which we further analyze by examining its behavior at specific entropy orders. 

To build a qualitative understanding of the TEE, 
we begin our discussion by examining the Hartley limit ($\alpha=0$).
When the circuit is composed of parameterized $\mathrm{SU}(4)$ unitaries (\figref{fig:overview}c.~inset) 
with randomly chosen parameters sampled from a Haar random distribution, 
the average Hartley entropy of a subsystem becomes
classically computable via graph-theoretic methods for large system sizes
\cite{skinnerMeasurementInducedPhaseTransitions2019}. 
This approach provides coarse-grained insight into the system's entanglement structure.

The main panel of \figref{fig:overview}c.~displays the ensemble average of $\overline{\TEE{0}}$,
(computed via the graph theoretic method) as a function of $D$, where $\TEE{\alpha}$ is a shorthand notation for 
the TEE of three contiguous equal-sized regions of size $n/4$, while the overline denotes the ensemble average. 
$\TEE{0}$ vanishes exactly up to $D\simeq n/8$, 
the minimal depth for the lightcone of a qubit at the center of the region to reach both of its neighboring subregions. 
This transition defines the \textit{quantum edge} depicted in \figref{fig:overview}, 
and identifies the region about a point 
where scrambling begins to cause the TEE to deviate from zero  \cite{kuriyattilOnsetScramblingDynamical2023}. 

As $D$ increases, $\TEE{0}$ rapidly decreases towards its fully scrambling limit of 
$-(n\ln 2)/2$ at $D = n/4$. 
This behavior remains qualitatively the same for larger $\alpha \geq 1$.
However, the depth $D$ required to reach the fully scrambled (Haar random) limit scales like
$\mathcal{O}(\alpha^{10} n)$ (for $\alpha > 1$) \cite{harrowApproximateUnitaryDesigns2023}. 
This scaling extends both the initial zero plateau and the subsequent decaying region beyond $D=n/4$. 
Consequently, the quantum edge of chaos can be identified as the point at which TEE starts to deviate from $0$. 

\begin{figure*}[t!]
   \includegraphics[scale=1.0]{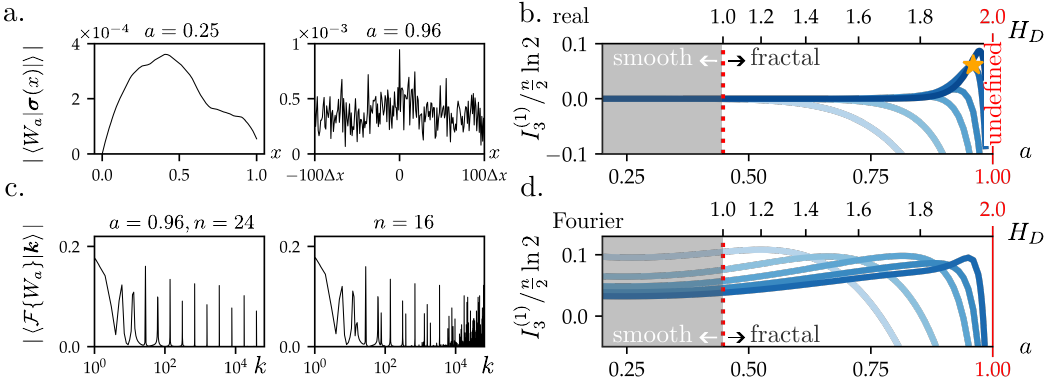}
   \caption{
      \textbf{Fractality vs quantum sparsity: the Weierstrass function.}
      {\bf a.}~The Weierstrass function in two different regimes: 
      (left) the smooth regime with $a=0.25<1/b$ and (right) the highly fractal regime with $a=0.96>1/b$ for $b=\sqrt{5}$
      at $n=24$ ($\Delta x = 2^{-24}$). 
      {\bf b.}~The TEE $I^{(1)}_3$ of the Weierstrass function across the transition from smooth ($a<1/b$) 
      to fractal ($ 1/b < a < 1$)
      for different system sizes $n=8,12,\ldots,28$ (light to dark) with $b=\sqrt{5}$. 
      The corresponding Hausdorff (fractal) dimension \cite{hunt_hausdorff_1998} 
      for a given $a$ is indicated on the top axis. 
      {\bf c.}~The absolute value of the amplitudes of $\ket{W_{a}}$ 
      for $a=0.96$ (indicated by the yellow star in a) at $n=24$  (top) and $n=16$ (bottom)
      in Fourier space, where $\ket{\mathcal{F}\{ W_{a} \}}=\mathcal{F}\ket{W_a}$ 
      obtained after applying QFT, $\mathcal{F}$, and  $\ket{\bm{k}}$ is the basis state in Fourier space (see Methods). 
      {\bf d.}~TEE of the Weierstrass function in Fourier space.
      The TEE is computed for $\ket{\mathcal{F}\{ W_{a} \}}$ for the same $a$ and $n$ as in panel b,
      with the exception of $n=28$ where the computational cost is prohibitive.
      \label{fig:weierstrass}} 
\end{figure*}

While Hartley entropy provides a coarse-grained view of the behavior of TEE with respect to the circuit depth, 
we now turn to $\alpha=1$ to examine the von Neumann entropy.
As the quantum counterpart of the classical Shannon entropy, $\TEE{1}$ provides a refined measure, 
allowing us to investigate the extent to which the quantum edge of chaos encompasses structured data 
that we may reasonably encounter in a computation.
To this end, we analyze the state $\ket{W_a}$ which amplitude-encodes the Weierstrass function (WF), 
$W_a(x)=\sum_{n=0}^{\infty} a^n \sin(b^n \pi x)$ ($0<b$) in the domain $0\leq x<1$. 
In this encoding, the domain is discretized into $2^n$ equidistant points $x_i = i \cdot 2^{-n}$, 
the function values at $x_i$, $W_a(x_i)$, are then assigned to the amplitude of the $i$\textsuperscript{th} 
computational basis state $\ket{\bm{\sigma}(x_i)}$ according to the bit string associated with the 
binary representation of $i$, up to a normalization factor. 
Here the computational basis state consists of local computational basis states $\ket{0}$ and $\ket{1}$, 
the $+1$ ($-1$) eigenstate of the Pauli-$Z$ operator; the $n$-fold tensor product of these naturally generates $2^n$ bit strings. 

The WF is a prototypical example of a function with tunable smoothness:
it is smooth  in the regime $0<a < 1/b$ with a bounded spectrum (\figref{fig:weierstrass}a., left),
while for $1/b \leq a < 1$, it develops a non-differentiable fractal character (\figref{fig:weierstrass}a., right).
Hence, the WF can serve as a test case in which the state encoding it interpolates continuously 
between smooth and irregular.
\figref{fig:weierstrass}b.~shows $\TEE{1}$ of $\ket{W_a}$ for varying system sizes. 
In the smooth and continuous regime, 
the TEE vanishes, as expected from its sparse, periodic, spatially correlated structure.
For $1/b<a$, however, a strong $n$ dependence of the TEE is observed.%

To connect this observation to the notion of quantum sparsity, we perform a quantum Fourier transform (QFT), $\mathcal{F}$,
to obtain $\ket{\mathcal{F}\{W_a\}}$. 
As indicated by its construction, the state is sparse in Fourier space 
as a function of integer wavenumber $k$, 
as shown in \figref{fig:weierstrass}c.~for $n=24$ (left) and $n=16$ (right).
However, for small $n$, the accumulation of amplitudes in large wavenumbers $k$ 
and the emergence of aliasing-induced peaks delocalize the state, shifting it toward a regime associated with negative TEE.
Furthermore, in \figref{fig:weierstrass}d., we plot the TEE of $\ket{\mathcal{F}\{W_a\}}$ 
for direct comparison with the real-space results shown in \figref{fig:weierstrass}b. 
We observe that the scrambling effect of an inverse QFT results in a minor suppression of TEE, 
while the overall behavior remains consistent with the predictions from our random state analysis. 
This correspondence is grounded in the formal relationship between TEE and classical sparsity: 
a random vector with sparsity $\mathcal{K}$ (defined as the number of nonzero-amplitude basis states) yields a non-negative TEE 
$\TEE{1}$ only if $\mathcal{K} < 2^{n/3}$ (see Methods for the derivation). 
This result is also consistent with Ref.~\cite{chenQuantumFourierTransform2023}, which indicates that 
the (inverse) QFT largely preserves the global entanglement structure of a quantum state.
Overall, these results establish TEE as a suitable indicator of quantum sparsity and a possible indirect measure of the underlying classical sparsity.

This behavior is formalized in the QNSST, which stipulates 
a threshold $q_c=\log_2(\pi/\lambda_{\min})$ 
for the required number of qubits to faithfully amplitude-encode a function. 
Here, $\lambda_{\min}$ is the wavelength corresponding to the highest-frequency Fourier component
of the function that is defined in the domain $x\in[0,1)$.
In the Methods section, we derive this bound from an exact representation of a state 
that amplitude-encodes $\sin(2 \pi x/\lambda)$. 
We show that the RDO of the $q$\textsuperscript{th} qubit, $\rho_q$, 
in the thermodynamic limit ($n \to \infty$) exponentially converges towards the Riemann interpolation of the function. 
Furthermore, unlike the classical sampling theorem, this threshold $q_c(\lambda)$ corresponds 
to a grid with resolution $\lambda / \pi$.
By linearity, this analysis holds for any function with a Fourier spectrum with compact support, 
and explains the vanishing TEE in the smooth region of \figref{fig:weierstrass}b. 

QNSST has profound implications for the quantum resources necessary to amplitude-encode a function onto a quantum state.
Here, we focus on the scaling of the required circuit depth for ansatz circuits that have the structure depicted in 
\figref{fig:overview}a. 
To encode a function $f(x)$ into a state $\ket{f(x)}$ 
with a well-defined Fourier spectrum characterized by a minimum wavelength $\lambda_{\min} $ 
and corresponding $q_c = \lceil \log_2 \pi/\lambda_{\min} \rceil$,
the required width $w$ of the circuit to reach the interpolation limit is of order $\mathcal{O}(q_c)$. 
For sufficiently well-behaved $\ket{f(x)}$, the depth $D$ of the circuit scales as (see Methods for details)
\begin{align}
   D=\mathcal{O}(q_c^2)=\mathcal{O}((\log_2 \pi/\lambda_{\min})^2), 
\end{align}
were the coefficients and the subleading terms depend on the specific choice of parameterization. 
Hence, the number of parameters required for the faithful representation of $f$ scales 
at worst as $\mathcal{O}(q_c^3)=\mathcal{O}((\log_2\frac{\pi}{\lambda_{\min}})^3)$, 
yielding the overall polylogarithmic scaling with $\lambda_{\min}$.

These results motivate us to apply TEE as a regularizer for quantum learning processes. 
For practical implementation, we restrict the order of \renyientropy{} to $\alpha=2$, 
as it can be %
measured experimentally on currently available quantum hardware 
\cite{alvesMultipartiteEntanglementDetection2004,
daleyMeasuringEntanglementGrowth2012,islamMeasuringEntanglementEntropy2015,
brydgesProbingRenyiEntanglement2019,hokeMeasurementinducedEntanglementTeleportation2023}.
We first investigate how the variance of the gradient of $\TEE{2}$ with respect to the parameter in the first layer 
depends on the total depth $D$. 
To do so, we restrict the structure to parameterized rotation gates acting on each qubit, 
where the rotation axis is drawn randomly from the Pauli operators and followed by a controlled-NOT gate 
(\figref{fig:optimization}a.). 
The main panel of \figref{fig:optimization}b.~shows the variances of the gradients of $\TEE{2}$, 
measured at the $\lfloor n/8 \rfloor$\textsuperscript{th} qubit at the depth $D=1$, 
$\theta_{1,\lfloor n/8 \rfloor} \in \bm{\theta}$.
The magnitude of the peaks %
is almost constant across the different system sizes, and their positions increase linearly with $n$. 
This suggests a cliff-like structure where the TEE, which is initially $\TEE{2}=0$, rapidly
decreases to the chaotic limit, analogous to our observations in \figref{fig:overview}c.~for $\alpha=0$. 

While the exponential gradient decay for $D \gg n$ might suggest a barren plateau, 
this is not an inherent limitation for local circuits.
Due to causality, penalizing entanglement at a small, fixed subregion size is sufficient to suppress entanglement 
at larger scales for the circuit discussed in this article, thereby maintaining trainable gradients. 
We propose a simple function evaluated on $\ket{\psi_D(\bm{\theta})}$ as a penalty term for VQA optimizations
\begin{align}
   \mathcal{C}_{\mathrm{TEE}}(\bm{\theta}) = \frac{1}{|\Omega|}\sum_{\{A,B,C\} \in \Omega} |\TEE{2}(A,B,C)_{\bm{\theta}}|,
   \label{eq:penalty}
\end{align}
where $\TEE{2}(A,B,C)_{\bm{\theta}}$ denotes the TEE of the parameterized output state,
and $\Omega$ is a set of triplets of disjoint, local subregions $\{A,B,C\}$.
This yields the regularized cost function
\begin{align}
   \mathcal{C}_{\mathrm{reg}}(\bm{\theta}) = \mathcal{C}(\bm{\theta}) + \gamma \mathcal{C}_{\mathrm{TEE}}(\bm{\theta}),
\end{align}
where $\mathcal{C}(\bm{\theta})$ is the bare cost function of the problem and $\gamma\geq0$ 
is a Lagrange multiplier. 

\begin{figure*}[t!]
   \includegraphics[scale=1.0]{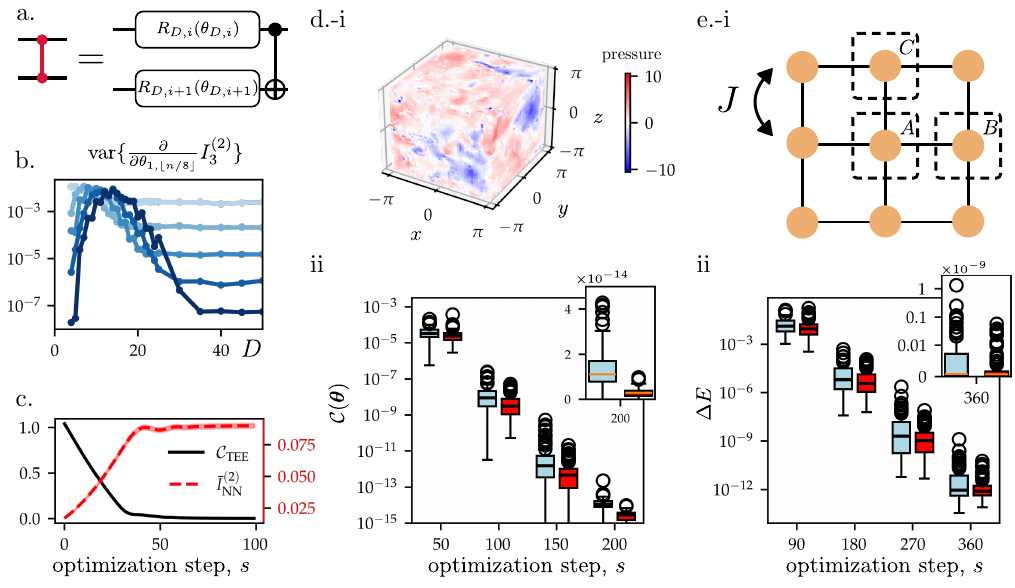}
   \caption{
      {\bf Regularized VQA optimization of brickwork ansatz circuits.}
      {\bf a.} The building blocks of the ansatz circuit at depth $D$ acting on qubits $i$ and $i+1$ 
      used for producing the results in this figure. 
      The circuit consists of parameterized rotation gates, denoted as $R_{D,i}(\theta_{D,i})$ and $R_{D,i+1}(\theta_{D,i+1})$, 
      which are applied to qubits $i$ and $i+1$ respectively, followed by a control-NOT gate, 
      to form a variational quantum state. Here $\theta_{D,i}$ controls the angle of rotation for the gate at depth $D$ 
      at qubit $i$. 
      For each rotation gate $R$, the axis of rotation is chosen randomly from the $x$, $y$, and $z$ axes of the Bloch sphere. 
      {\bf b.}~The gradient of $\TEE{2}$ 
      with respect to the $\lfloor n/8 \rfloor$-th
      parameter at $D=1$ is shown as a function of the total depth of the circuit $D_{\text{tot}}$
      for $n=8$, $12$, $16$, $20$, and $24$ (light to dark).
      The result is the average over 100 random initializations  
      (see Methods, error bars are too small to be visible). 
      {\bf c.}~The evolution of $\mathcal{C}_{\mathrm{TEE}}$ (black) and $\bar{I}^{(2)}_{\mathrm{NN}}$ (red) 
      over optimization steps $s$ for $n=9$ and depth $D=180$.
      The circuit parameters are randomly initialized, and the quantities are calculated on the 
      output state at depth $D$, $\ket{\psi_D(\bm{\theta})}$ (see Methods for details).
      {\bf d.}~Result for variational encoding of a turbulent flow field: 
      {\bf i.}~A size $256\times256\times256$ snapshot of the pressure field taken from the 
      isotropic turbulent flow dataset (see Methods). 
      It is used as a reference to create an $n=9$ qubit amplitude-encoded reference state. 
      {\bf ii.}~The distribution of the cost function at different optimization steps $s=50$, $100$, $150$, and $200$
      for $\gamma_0=0.1$ (red) and $\gamma_0=0$ (blue) with $\beta=0.5$ for $n=9$.
      Inset: the magnification of the cost function value at $s=200$. 
      {\bf e.}~Regularized ground state search in AF-2D-NNH:
      {\bf i.}~A sketch of the AF-2D-NNH model and the choice of L-shaped triplets $A$, $B$, and $C$, 
      for computing the TEE for the middle qubit. 
      {\bf ii.}~The distribution of the distance from the true solution $\Delta E$ at different optimization steps 
      $s=90$, $180$, $270$, and $360$ for $\gamma_0=100$ (red) and $\gamma_0=0$ (blue) with $\beta=0.5$ for $n=9$. 
      Inset: the magnified view of the distribution at $s=360$. 
      The y-axis scale is linear for $\mathcal{C}(\bm{\theta})\leq 10^{-11}$ and logarithmic for higher values.
      The results for d.~and e.~are obtained by sampling 100 realizations (See Methods for the details of optimizations). 
   \label{fig:optimization}}
\end{figure*}

We start by analyzing the effect of the regularizer on the entanglement structure of a randomly initialized 
ansatz circuit discussed previously. 
Thus, we perform the optimization solely with the regularizer $\mathcal{C} = \mathcal{C}_{\mathrm{TEE}}(\bm{\theta})$ for 
$n=9$, $D=180$, and with the randomly initialized $\bm{\theta}$ (see Methods). 
We then the set $\Omega=\Omega_2$, where $\Omega_2$ is a set of all triplets of contiguous local subregions 
in one dimension with periodic boundary conditions, each of size $n_s=2$.
For this choice of $\Omega$, 
the evaluation of the regularizer term $\mathcal{C}_{\mathrm{TEE}}$ adds a measurement cost of $\mathcal{O}(n_sn)$
for the Hong-Ou-Mandel-type purity-detection method
\cite{alvesMultipartiteEntanglementDetection2004,daleyMeasuringEntanglementGrowth2012,islamMeasuringEntanglementEntropy2015}.
This cost can be further reduced $\mathcal{O}(n)$ to with the randomized measurements
\cite{brydgesProbingRenyiEntanglement2019,hokeMeasurementinducedEntanglementTeleportation2023}. 
In \figref{fig:optimization}~b., we show the evolution of $\mathcal{C}_{\mathrm{TEE}}(\bm{\theta})$ (black)
and the average nearest-neighbor mutual information $\bar{I}^{(2)}_{\mathrm{NN}} = (1/n)\sum_i I^{(2)}(i,i+1)$ 
with periodic boundary condition (red). 
The results clearly show that entanglement between geometrically close qubits is enhanced concomitantly with the decay of TEE. 
The emergence of this structure suggests that the system operates near maximal computational capacity, 
consistent with the classical hypothesis of computing at the edge of chaos \cite{langtonComputationEdgeChaos1990,sole1996phase}.

To confirm this, we perform regularized VQA on two use cases.
We encode the scalar pressure field of a turbulent flow from Refs.~\cite{Perlman2007,Li2008,Yeung2015} 
shown in \figref{fig:optimization}d.~i.
We choose this example as it represents one of the most complex cases that may emerge as a VQA use case in 
the context of computational fluid dynamics or more generally, for solving nonlinear PDEs
\cite{jakschVariationalQuantumAlgorithms2023,bengoecheaVariationalQuantumAlgorithms2024}. 
An approximate representation of this state is generated by optimizing the fidelity between the exactly encoded state
and the variational state $\ket{\psi_D(\bm{\theta})}$ by setting 
$\mathcal{C}(\bm{\theta}) = 1 - |\braket{\psi( \bm{\theta})|\Pi}|^2$, where $\ket{\Pi}$ is a reference 
state amplitude-encoding the scalar pressure field of an isotropic turbulent flow (see Methods).
We perform optimization with the same $\Omega=\Omega_2$. 
Naturally, the regularized cost function $\mathcal{C}_{\mathrm{reg}}(\bm{\theta})$ 
converges only to the true minimum of the bare cost function 
$\mathcal{C} (\bm{\theta})$, if the TEE of this minimum vanishes exactly.
However, based on our previous discussion, we must not assume that this condition is fulfilled a priori. 
Therefore, for all training discussed subsequently, we reduce the penalty term exponentially as the optimization step 
proceeds (See Methods).
This strategy of using a step-dependent regularizer is analogous 
to established continuation techniques in $\ell_1$-regularized minimization \cite{hale2007fixed}. 
It allows the VQA to converge to the minimum of $\mathcal{C} (\bm{\theta})$ after the constrained optimization has generated 
a sparse approximate solution near the true minimum, while avoiding barren plateaus. 
In \figref{fig:optimization}d.~ii, we present the distribution of the cost function over different trajectories 
for optimization results with regularization (red) and without regularization (blue).
We observe that in the later stages of the optimization, the regularized optimization converges 
to the global minimum with high fidelity.
In contrast, the optimization without regularization shows slower convergence with 
outliers deviating by one to two orders of magnitude from the worst-case trajectory observed for the regularized optimization.

A similar trend also emerges in the ground state search for the antiferromagnetic, two-dimensional, nearest-neighbor Heisenberg 
(AF-2D-NNH) model (\figref{fig:optimization}e.~i) on the square lattice, 
a prototypical model for quantum magnetism, with the Hamiltonian
$H=J\sum_{<i,j>}(\sigma^{x}_i\sigma^{x}_{j}+\sigma^{y}_i\sigma^{y}_{j}+\sigma^{z}_i\sigma^{z}_{j}) + h\sum_i \sigma^{z}_i$,
where $<\cdot,\cdot>$ denotes the nearest-neighbor sites and $\{\sigma^{x}_i,\sigma^{y}_i,\sigma^{z}_i\}$ 
are Pauli operators acting on qubit $i$. 
We solve for the ground state by 
optimizing the cost function $\mathcal{C}(\bm{\theta})=\braket{\psi_D(\bm{\theta})|H|\psi_D(\bm{\theta})}$,
for the model with parameters $J=1$ and $h/J=2$, where 
the ground state features long-range N\'eel order with a nontrivial entanglement structure
\cite{giamarchi2003quantum,songEntanglementEntropyTwodimensional2011}. 
For this optimization, we choose $\Omega$ as the set of all L-shaped triplets of nearest-neighbor sites,
as illustrated in \figref{fig:optimization}e.~i.
Hence, the constraint term not only mitigates barren plateaus
but also imposes a geometric constraint on the optimization of a problem with a geometrically incommensurate ansatz circuit. 
Successfully preparing such a state with a VQA therefore demonstrates the robustness of our approach 
in tackling problems defined by complex, geometrically nonlocal quantum correlations. 
This makes it an ideal testbed for validating the efficacy of our new optimization method against 
a physically relevant and computationally hard problem.
We perform the VQE with the same step-dependent regularizer strength (see Methods).
In \figref{fig:optimization}e~ii., we plot the deviation from the true ground state energy 
$\Delta E = \mathcal{C}(\bm{\theta}) - E_g $, where $E_g$ is the true ground state energy of the model. 
We observe again faster and more controlled convergence for optimization with regularization
(red) in comparison to the case without (blue).
A suppression of outliers by one to two orders of magnitude is also observed.

These drastic improvements shown in these two examples from two widely separated fields of physics show that the penalty term 
(\eqnref{eq:penalty}) effectively reshapes the optimization landscape. 
This behavior supports our hypothesis that the true optimum lies within the \textit{uantum edge}, a critical regime, 
bordering between a fertile valley and a barren plateau. 
Furthermore, the theoretical, analytical, and benchmark results discussed in this article 
shed light on a potential unifying explanation for various existing barren plateau mitigation techniques,
ranging from specialized initializations 
\cite{grantInitializationStrategyAddressing2019,mhiriUnifyingAccountWarm2025}
and localization \cite{parkHardwareefficientAnsatzBarren2024}
to local optimization \cite{huangGenerativeQuantumAdvantage2025}. 
Our results suggest that these diverse methods converge on a common principle discussed in this article: 
avoiding barren plateaus by creating highly expressive, shallow-circuit-equivalent structures and mitigating
scrambling by restricting lightcone spreading.

Moving forward, the principle of quantum sparsity and the associated QNSST offer a rigorous framework 
for circuit architecture design. 
These principles likely extend beyond VQA to any quantum protocol 
where structured information must be protected from uncontrolled scrambling. 
This work establishes a link between quantum information geometry and computational tractability, 
providing a roadmap for achieving practical quantum advantage in the noisy intermediate-scale quantum era and beyond.

\bibliography{bibs.bib}

\clearpage
\section{Methods}

\pheader{TEE of a random $\mathcal{K}$-sparse state}
Firstly, we show that the threshold classical sparsity $\mathcal{K}$ at which $\TEE{3}$ becomes negative is $2^{n/3}$
for an $n$-qubit random sparse state. 
Here, $TEE{\alpha}$ denotes the $\TEE{\alpha}(A,B,C)$, formulated with the von Neumann entropy,
where the regions $A$, $B$, and $C$ are chosen as contiguous regions of equal size $|A|=|B|=|C|=n/4$, consistent with the main text.
Adopting definitions from compressed sensing \cite{abo-zahhadCompressiveSensingAlgorithms2015}, 
we define the classical sparsity $\mathcal{K}$ of a quantum state as the number of non-zero amplitude basis states 
in a given representation, as briefly mentioned in the main text.
We then construct a random $\mathcal{K}$-sparse pure state $\ket{\mathcal{K}}$ as
\begin{align}
\ket{\mathcal{K}} = \sum_{i\in \mathcal{S}} c_i \ket{i}
\end{align}
where $\mathcal{S} \subset \{0, 1, \dots, 2^n - 1\}$ is a random subset of basis state indices with cardinality 
$|\mathcal{S}|=\mathcal{K}$. 
Each index $i$ corresponds to a computational basis state $\ket{i} \equiv \ket{b_{n-1} \dots b_0}$, 
where the bit string $b_{n-1} \dots b_0$ represents the binary expansion of the integer $i$. 
The coefficients $c_i$ are random complex numbers with uniform magnitude, $|c_i| = 1/\sqrt{\mathcal{K}}$, 
to ensure the normalization $\braket{\mathcal{K}|\mathcal{K}} = 1$, and differ only by their stochastic phases. 
Thus, $\ket{\mathcal{K}}$ effectively represents the amplitude encoding of a classical array with sparsity $\mathcal{K}$
and elements of uniform magnitude. 

When tracing out the environment of a given region, $X\in\{A,B,C\}$, 
the contributions of the cross-terms of the form $c_ic_j^*\braket{i|j}$ ($i\neq j$) 
vanish exponentially with the system size, scaling as $1/2^{n-|X|}$. 
Therefore, the reduced density operator of $X$ is dominated by the diagonal components 
while off-diagonal terms vanish exponentially with $n$. 

We ignore these vanishing terms in this analysis, and assume $2^{n/4} \ll \mathcal{K}$ and $\mathcal{K} \ll 2^{n/2}$.
Under these assumptions, we obtain the von Neumann entropy of $X$ as $S^{(1)}_X \approx |X|\ln 2$
where $|X| \ll \log_2 \mathcal{K}$, for subregion $X \in \{A,B,C\}$. 
For a joint subpartition $XY=X\cup Y$ ($Y \in \{A,B,C \}\backslash \{ X \}$, 
the entropy $S^{(1)}_{XY}$ generally approaches $\ln \mathcal{K}$.
Thus, the $\TEE{1}$ is
\begin{align}
   \TEE{1}(0,n/4) \approx  n\ln 2 - 3\ln \mathcal{K}
\end{align}
Consequently, the condition on $\kappa$ for $0<\TEE{1}$ is 
\begin{align}
   \mathcal{K} < 2^{\frac{1}{3}n}. \label{eqn:boundkappa}
\end{align}
Thus, a positive TEE is a signature of an underlying sparse structure in a quantum state.
In \appref{app:numksparseTEE} we provide numerical results for a random state with sparsity $\mathcal{K}$, and
equal amplitudes; those results demonstrate a clear transition in TEE from positive to negative near $\mathcal{K} = 2^{n/3}$,
consistent with the theoretical expectations. 
\\

\pheader{Exact amplitude encoding of a sine function}
We provide an exact representation of a single-qubit reduced density operator of the state $\ket{\psi_{\lambda}^n}
\propto\sum_{i}\sin(\frac{2\pi}{\lambda} x_i)\ket{\bm{\sigma}(x_i)} $
defined over the domain $0\leq x < 1$. 
Here, $x_i$ are the equidistant grid points in the domain and $\lambda$
is the wavelength, as discussed in the main text.

In the form of a Matrix Product State (MPS) \cite{oseledetsConstructiveRepresentationFunctions2013}, 
$\ket{\psi_{\lambda}^n}$ is written as
\begin{align}
   \ket{\psi_{\lambda}^n} = \sum_{\bm{\sigma}}
   \sum_{\bm{\kappa}} M^{(\sigma_0)}_{\kappa_0}
   M^{(\sigma_1)}_{\kappa_0,\kappa_1}
   \ldots
   M^{(\sigma_{n-1})}_{\kappa_{n-2}}\ket{\bm{\sigma}(x_i)},
\end{align}
where $M^{\sigma_k}_{i,j}$ denotes a rank-three tensor and summation is taken over all the possible
computational basis states $\bm{\sigma}=\{ \bm{\sigma}(x_0),\bm{\sigma}(x_1),\cdots,\bm{\sigma}(x_{2^{n-1}})\}$ 
and bond indices $\bm{\kappa} = \{ \kappa_0,\kappa_1,\cdots\kappa_{n-1}\}$.
Consistent with the main text, each basis state $\ket{\bm{\sigma}} = \bigotimes_{k=0}^{n-1} \ket{\sigma_k}$ 
is identified as the $i$\textsuperscript{th} computational basis state $\ket{\bm{\sigma}(x_i)}$, 
where the bit string $(\sigma_0, \dots, \sigma_{n-1})$ is the binary representation of the integer $i$ 
that defines the grid point $x_i = i \cdot 2^{-n}$.

The tensors $M$ are explicitly given by
\begin{align}
   M_{\kappa_0}^{(0)} (\tau) &= \begin{pmatrix} \sin\left(\frac{k}{2}+\tau\right) & \cos\left(\frac{k}{2}+\tau\right) \end{pmatrix}, 
   M_{\kappa_0}^{(1)} (\tau) = \begin{pmatrix} \cos(\tau) & \sin(\tau) \end{pmatrix} \nonumber\\
   M_{\kappa_{q-1},\kappa_{q}}^{(0)} (\tau) &= 
   \begin{pmatrix}
       \cos\left(\frac{k}{2^{q+1}}+\tau\right) & \sin\left(\frac{k}{2^{q+1}}+\tau\right) \\
      -\sin\left(\frac{k}{2^{q+1}}+\tau\right) & \cos\left(\frac{k}{2^{q+1}}+\tau\right)
   \end{pmatrix},
   M_{\kappa_{q-1},\kappa_{q}}^{(1)}  (\tau)= 
   \begin{pmatrix}
      \cos(\tau) & \sin(\tau) \\
      -\sin(\tau) & \cos(\tau)
   \end{pmatrix}\nonumber\\
   M_{\kappa_{n-2}}^{(0)}  (\tau)&= 
   \begin{pmatrix} 
      \cos\left(\frac{k}{2^n}+\tau\right) \\ \sin\left(\frac{k}{2^n}+\tau\right) 
   \end{pmatrix},
   M_{\kappa_{n-2}}^{(1)}  (\tau)= \begin{pmatrix} \sin(\tau) \\ \cos(\tau) \end{pmatrix}. 
\end{align}
where $k=2\pi/\lambda$ is the wavenumber and $q \in \{1, \dots, n-2\}$ denotes the site index. 
To encode a sine function with a non-zero $\phi$, we set the tensor parameter of one chosen site to
$\tau=\phi$ while setting  $\tau=0$ otherwise. 
Unless otherwise explicitly stated, we set $\tau=0$ for the remainder of this section.

In the computational basis $\{ \ket{0}, \ket{1}\} $, the reduced density operator (RDO) of qubit $q$, obtained by tracing out all other qubits $\bar{q} \neq q$, is given by 
$\rho_q=\mathrm{Tr}_{\bar{q}} \ket{\Psi_{\lambda}^n}\bra{\Psi_{\lambda}^n}$.
This operator is obtained by contracting
the MPS tensors over all physical indices $\sigma_k$, except for $\sigma_q$, as well as all bond indices. 
In the limit of $n\to\infty$, we obtain the following expression
\begin{align}
   \rho_q^{\infty} &= \lim_{n\to\infty} \sum_{\bm{\kappa},\bm{\kappa'},\bm{\sigma}\neq \sigma_q} 
   M^{(\sigma_0)}_{\kappa_0}\cdots M^{(\sigma_{q})}_{\kappa_{q-1},\kappa_{q}}(\phi) \cdots M^{(\sigma_{n-1})}_{\kappa_{n-2}}
   (M^{(\sigma_0)}_{\kappa_0'})^*\cdots(M^{(\sigma_{q})}_{\kappa_{q-1}',\kappa_{q}'}(\phi) )^* \cdots (M^{(\sigma_{n-1})}_{\kappa_{n-2}'})^*  \nonumber\\
          &= 
   \begin{pmatrix}
      \frac{2k-\sec(2^{-q-1} k)
      \left(\sin((2+2^{-q-1})k+2\phi) - \sin(2^{-q-1}k + 2\phi)\right)}{4 k - 2\sin(2(k+\phi)) + 2\sin(2\phi)} & 
      \frac{2k\cos(2^{-q-1} k) - \sec(2^{-q-1}k)\left(\sin(2(k+\phi)) - \sin(2 \phi)\right)}
      {4 k - 2\sin(2(k+\phi)) + 2\sin(2\phi)} \\
      \frac{2k\cos(2^{-q-1} k) - \sec(2^{-q-1}k)\left(\sin(2(k+\phi)) - \sin(2 \phi)\right)}
      {4 k - 2\sin(2(k+\phi)) + 2\sin(2\phi)} &
      \frac{2k-\sec(2^{-q-1} k)
      \left(\sin((2-2^{-q-1})k+2\phi) - \sin(2\phi-2^{-q-1}k)\right)}{4 k - 2\sin(2(k+\phi)) + 2\sin(2\phi)}
   \end{pmatrix},
   \label{eq:singleqrdmat}
\end{align}
where $M^{*}$ denotes the elementwise complex conjugation of $M$. Derivation is provided in \appref{app:derivrhoinf}.

Based on the expression, for $2^{-q-1}k < 1$, 
the reduced density operator converges exponentially towards that of the state $\ket{+}\bra{+}$ as $q$ increases.
By rearranging the inequality $2^{-q-1}k < 1$, we obtain the threshold number of qubits for resolution, 
$q_c(\lambda)=\log_2 \pi/\lambda$, as stated in the main text.

Furthermore, in the linear interpolation limit, an ($n+1$)-qubit state is generated by a linear interpolation of an $n$-qubit state as
\begin{align}
\ket{\Psi_{\lambda}^{n+1}} \propto \left(\mathbb{I}^{\otimes n+1}+\frac{\mathbb{I}^{\otimes n}+\mathbb{A}}{2}\otimes\sigma^{x}_{n+1}\right) 
   \ket{\Psi_{\lambda}^{n}}\ket{0}_{n+1},
\end{align}
where operator $\mathbb{A}$ is a unitary operator that 
shifts a basis state by one grid size, $\mathbb{A} \ket{\bm{\sigma}(x_i)}=\ket{\bm{\sigma}(x_i-2^{-n})}$ 
((assuming periodic boundary conditions) and $\mathbb{I}$ is the identity acting on the local one-qubit Hilbert space.

Assuming that $\ket{\Psi_{\lambda}^{n}}$ is exact, the reduced density operator of the ($n+1$)\textsuperscript{th}-qubit,
$\rho_{n+1}$, is calculated as
\begin{align}
   \rho_{n+1} &= \frac{1}{2}
   \begin{pmatrix}
      1 & \frac{1}{2} + \frac{1}{2\mathcal{N}^2}\sum_{i=0}^{2^{n}-1} 
      \sin(\frac{2\pi}{\lambda} x_i)\sin(\frac{2\pi}{\lambda}(x_i+2^{-n})) \\ 
      \frac{1}{2} + \frac{1}{2\mathcal{N}^2} \sum_{i=0}^{2^n-1} \sin(\frac{2\pi}{\lambda}x_i)\sin(\frac{2\pi}{\lambda}(x_i+2^{-n})) &
      \frac{1}{2} + \frac{1}{2\mathcal{N}^2} \sum_{i=0}^{2^n-1} \sin(\frac{2\pi}{\lambda}x_i)\sin(\frac{2\pi}{\lambda}(x_i+2^{-n})) 
   \end{pmatrix}
   \nonumber\\
   &\approx
   \frac{1}{2} 
   \begin{pmatrix}
   1 & 1 + \mathcal{O}(2^{-n}) \\ 
   1 + \mathcal{O}(2^{-n}) &1 + \mathcal{O}(2^{-n})
\end{pmatrix} &(n_c \ll n),
\end{align}
where $\mathcal{N}=\sqrt{\sum_{i=0}^{2^n-1} \sin^2(\frac{2\pi}{\lambda} x_i)} = \mathcal{O}(2^{n/2})$
is the normalization constant, and the approximation in the second line holds for $\frac{2\pi}{\lambda}2^{-n}= 2^{1+n_c-n}\ll 1$. 
Clearly, the expression converges exponentially with $n$ towards the state $\ket{+}\bra{+}$, 
consistent with \eqnref{eq:singleqrdmat}.
Finally, in \appref{app:QNSSTmultiplefreq} 
we show that when the state has a continuous spectrum with band limit of $\lambda_{\min}$, 
the threshold qubit remains at $q_c(\lambda_{\min})$. 
\\

\pheader{Circuit depth scaling for brickwork ansatz circuit}
We analyze the circuit depth $D_n$ required to prepare an $n$-qubit state $\ket{f_{0,L}}$, 
which represents the amplitude-encoding of a function $f(x)$ defined on the domain $0 \leq x < L$. 
To achieve this, we employ recursive construction of the state. 
The target state can be expressed as a superposition of functions over its sub-domains as
$$\ket{f_{0,L}} = \mathcal{C}_{0,L/2}\ket{0}\ket{f_{0,L/2}} + \mathcal{C}_{L/2,L}\ket{1}\ket{f_{L/2,L}},$$
where $\ket{f_{0,L/2}}$ and $\ket{f_{L/2,L}}$ are $(n-1)$-qubit states on the former and latter half of the domain respectively.

The proposed construction method defines the VQA problem as finding a single ($n-1$)-qubit unitary 
$U(\bm{\theta})$ that maps two product states to two functions, 
such that $\ket{f_{0,L/2}} \approx U(\bm{\theta})\ket{0}^{\otimes n-1}$ and 
$\ket{g_{L/2,L}} \approx U(\bm{\theta})\ket{r}^{\otimes n-1}$.
Here, $\ket{g_{L/2,L}} \propto \ket{f_{L/2,L}} - \braket{f_{0,L/2}|f_{L/2,L}}\ket{f_{0,L/2}}$ 
is the normalized orthogonal component of the second function, and $\ket{r}$ is a single qubit state that ensures 
$\braket{f_{0,L/2}|g_{L/2,L}}=\braket{0|r}^{n-1}$. 

The full state, $\ket{f_{0,L}}$, is constructed by
first preparing an intermediate Greenberger-Horne-Zeilinger-like (GHZ-like) state of the form
$$\ket{\phi_{\mathrm{GHZ}}} = \mathcal{C}_{\ket{00}}\ket{0}\ket{0}^{\otimes (n-1)} 
+ \mathcal{C}_{\ket{1r}}\ket{1}\ket{r}^{\otimes n-1}$$
from the $\ket{0}^{\otimes n}$ state with the appropriate coefficients $\mathcal{C}_{\ket{00}}$ and $\mathcal{C}_{\ket{1r}}$. 
Subsequently, the VQA unitary $\mathbb{I} \otimes U(\bm{\theta})$ is applied to $\ket{\phi_{\mathrm{GHZ}}}$. 

The total circuit depth $D_n$ is derived from the combined cost of these steps:
First, the depth to create the state $\ket{\phi_{\mathrm{GHZ}}}$ using a brickwork ansatz circuit (as depicted in 
\figref{fig:overview}a.), which is of order $\mathcal{O}(n)$ \cite{baumerEfficientLongRangeEntanglement2024}, 
and second, the circuit depth contribution from $U(\bm{\theta})$. 
Now we assume that $n$ is sufficiently large ($q_c(\lambda_{\min}) < n$) 
such that $\ket{f_{0,L}}$ with a well defined smallest wavelength $\lambda_{\min}$ is well resolved,
the number of points of discontinuity in the function within the domain is independent of $n$,
and the entanglement structure and the underlying Fourier spectrum of the states in the subdomains 
($\ket{f_{0.L/2}}$ and $\ket{f_{L/2,L}}$) are similar.
With these assumptions, the depth of the ($n-1$)-qubit VQA circuit $U(\bm{\theta})$ 
is equivalent to the depth $D_{n-1}$ required to map one basis state to the target state. 
Summing these costs establishes the recursion relation and yields the derived scaling 
\begin{equation}
   D_{n} = \mathcal{O}(n) + D_{n-1} =\mathcal{O}(n^2).
\end{equation} 
Hence, the number of parameters required to encode $f(x)$ scales as $\mathcal{O}(n^3)$. 
Consequently, the circuit error also scales with the total number of gates,
which is equivalent to the total number of parameters in the case of a brickwork ansatz circuit,
namely $\mathcal{O}(n D_n) = \mathcal{O}(n^3)$ (\appref{app:errorscaling}).
Furthermore, to confirm the inductive relation and the derived scaling, we refer to \appref{app:weierstrassencoding}, 
where we provide the scaling of the minimum number of parameters that are required to achieve an infidelity 
 below a threshold $\epsilon_{\mathrm{th}}$ for the Weierstrass function $W_a(x)$.
We find the scaling to be consistent with or below $\mathcal{O}(n^3)$. 
\\

\pheader{Regularized optimization}
Finally, we explain the details of the Variational Quantum Algorithms (VQAs) investigated in the main text.
We analyze three distinct cases: optimization in the limit where the regularizer dominates, 
the preparation of a turbulent data-encoded state, and a variational quantum eigensolver (VQE) for 
finding the 
ground state of the 
antiferromagnetic, two-dimensional, nearest-neighbor Heisenberg
AF-2D-NNH) model.

To implement these tasks, we utilize a common variational framework based on a brickwork ansatz circuit, 
as illustrated in the circuit diagram of \figref{fig:overview}a.
In this architecture, $n$ qubits (represented 
by black horizontal lines) are evolved from an initial state $\ket{\psi_0}$ through a sequence of local entangling layers. 
Each entangling operation (red bars in \figref{fig:overview}a) corresponds to a block 
acting on a pair of nearest-neighbor qubits $(i, j)$ specific to that layer. Each block consists of two 
single-qubit parameterized rotation gates $R_{d,i}(\theta_{d,i})$ and $R_{d,j}(\theta_{d,j})$ acting on 
qubits $i$ and $j$ respectively, followed by a control-NOT ($\mathrm{CNOT}$) gate. The rotation gates 
are defined as $R_{d,i}(\theta_{d,i}) = \exp(-\mathrm{i}\frac{\theta_{d,i}}{2} \sigma^{\kappa_{d,i}}_i)$, 
where the rotation axis $\kappa_{d,i} \in \{x, y, z\}$ is chosen randomly at initialization for each 
qubit and depth. The brickwork ansatz circuit produces an output state $\ket{\psi_{\Dtot{}}(\bm{\theta})}=U(\bm{\theta})\ket{\psi_0}$, 
where the parameterized unitary $U(\bm{\theta}) = \prod_{d=1}^{\Dtot{}} U_d(\bm{\theta})$ is composed of total $\Dtot{}$ layers. 
To facilitate global information spreading, the layer structure alternates between 
acting on qubit pairs $(2i+1, 2i+2)$ for odd-numbered layers ($d=1, 3, \dots$) and $(2i, 2i+1)$ 
for even-numbered layers ($d=2, 4, \dots$), with periodic boundary conditions imposed by identifying 
the $(n+i)$\textsuperscript{th} qubit with the $i$\textsuperscript{th} qubit. Gradients of the TEE, 
formulated using the second-order \renyientropy{} $S^{(2)}_A = -\ln \Tr\{\rho_A^2\}$, 
are evaluated via the parameter-shift rule \cite{mitaraiQuantumCircuitLearning2018,vidalCalculusParameterizedQuantum2018}: 
$\frac{\partial}{\partial \theta_{d,i}}S^{(2)}_A = -\Tr\{(\rho_A^{+}-\rho_A^{-})\rho_A\} / \Tr\{\rho_A^2\}$, 
where $\rho^{\pm}_A$ are obtained from states generated by shifting the parameter $\theta_{d,i}$ by $\pm \pi/2$.

For the optimization in the limit where the regularizer dominates (\figref{fig:optimization}b.~and c.), 
we employ a circuit with $n=9$ qubits 
and a total depth of $\Dtot{}=180$. The cost function is defined solely by the TEE regularizer 
$\mathcal{C}(\bm{\theta}) = \mathcal{C}_{\mathrm{TEE}}(\bm{\theta})$, using a triplet set $\Omega = \Omega_2$ 
consisting of all triplets of contiguous local subregions, each of size 2 in one dimension under periodic boundary conditions. 
The initial state is 
$\ket{\sqrt[4]{Y}} = (\sqrt[4]{\sigma^{y}}\ket{0})^{\otimes n}$. 
The circuit parameters $\bm{\theta}$ are initialized randomly, and optimization is performed for 200 steps 
using the \texttt{ADAMW} optimizer from \texttt{TensorFlow} with default parameters.

In the preparation of the turbulent data-encoded state (\figref{fig:optimization}d.), we optimize the infidelity cost function 
$\mathcal{C}(\bm{\theta}) = 1 - |\braket{\psi( \bm{\theta})|\Pi}|^2$, where $\ket{\Pi}$ is a reference 
state amplitude-encoding the scalar pressure field of an isotropic turbulent flow. 
The ground-truth pressure field is obtained from the Johns Hopkins University isotropic turbulent flow dataset, \texttt{isotropic8192}, at snapshot time 1 \cite{Perlman2007,Li2008,Yeung2015}. 
The reference state $\ket{\Pi}$ is created by sampling the $256\times256\times256$ snapshot with an interval $\Delta x = 2\pi/2^n$ for $n=9$.
The domain is discretized 
into $2^n$ points for $n=9$ qubits with a total depth of $\Dtot=180$ layers. We employ the 
triplet set $\Omega = \Omega_2$ and a step-dependent regularizer $\gamma(s) = \gamma_0\beta^s$ with 
$\gamma_0=0.1$ and $\beta=0.5$, where $s$ is the optimization step count. The optimization is conducted 
for 200 steps using the Conjugate Gradient (\texttt{CG}) implementation in \texttt{SciPy} \cite{2020SciPy-NMeth}
to ensure high-precision convergence in the final stage, 
starting from the initial state $\ket{\psi_0}=\ket{\sqrt[4]{Y}}$ with the first 6 layers initialized with random parameters.

For the VQE task (\figref{fig:optimization}e.), 
we seek the ground state of the AF-2D-NNH model on a $3\times3$ square lattice ($n=9$) 
with the Hamiltonian $H=J\sum_{<i,j>}(\bm{\sigma}_i \cdot \bm{\sigma}_j) + h\sum_i \sigma^{z}_i$ (for $J=1, h/J=2$).
The cost function is the energy expectation value $\mathcal{C}(\bm{\theta})=\braket{\psi(\bm{\theta})|H|\psi(\bm{\theta})}$, 
regularized with $\gamma_0=100$ and $\beta=0.5$. We choose $\Omega$ as the set of all L-shaped
triplets of nearest-neighbor sites to reflect the two-dimensional connectivity of the lattice (\figref{fig:optimization}e.~i). 
The optimization is performed for 360 steps using the \texttt{CG} optimizer, utilizing $D_{\text{tot}}=180$ layers 
and starting from the initial state $\ket{\sqrt[4]{Y}}$ with the first 6 layers initialized with random parameters.
\\

\pheader{Data availability}
The dataset containing the source data for this manuscript is available at \cite{hashizumedataset2026}.
All other data that support the plots within this paper and other findings of this study are available from the corresponding author upon reasonable request.

\section{Acknowledgements}
DJ, TH, and FS are partially funded by the Cluster of Excellence `Advanced Imaging of Matter' of the Deutsche Forschungsgemeinschaft (DFG)|EXC 2056- project ID390715994.
DJ and TH acknowledge support by the European Union’s Horizon Programme (HORIZON-CL42021DIGITALEMERGING-02-10) Grant Agreement 101080085 QCFD.
DJ acknowledges DFG project „Quantencomputing mit neutralen Atomen“ (JA 1793/1-1, Japan-JST-DFG-ASPIRE 2024), 
and the Hamburg Quantum Computing Initiative (HQIC) project EFRE. 
The EFRE project is co-financed by the ERDF of the European Union and by the
``Fonds of the Hamburg Ministry of Science, Research, Equalities and Districts (BWFGB)''.
Results were obtained using the PHYSnet computational cluster based at University of Hamburg 
and the HPC-cluster Hummel-2 at University of Hamburg. 
FS acknowledges support from the research unit `FOR5750: OPTIMAL' - project ID 531215165. 
The Hummel cluster was funded by Deutsche Forschungsgemeinschaft (DFG, German Research Foundation) – 498394658.
ZW acknowledges the support of the Swedish Research Council VR (Grant 2022–06176).
TH acknowledges Martin Stieben for support in obtaining the numerical results and Gregory S.~Bentsen for helpful discussions. 

\appendix 
\clearpage
\section{Numerical result for random $\mathcal{K}$-sparse state}
\label{app:numksparseTEE}
In this appendix, we show the numerical result for the TEE of $\mathcal{K}$-sparse quantum states $\ket{\mathcal{K}}$. 
As stated in the Methods, we construct a random $\mathcal{K}$-sparse pure state $\ket{\mathcal{K}}$ as
\begin{align}
\ket{\mathcal{K}} = \sum_{i\in \mathcal{S}} c_i \ket{i}
\end{align}
where $\mathcal{S} \subset \{0, 1, \dots, 2^n - 1\}$ is a random subset of $\mathcal{K}$ unique integers with cardinality 
$|\mathcal{S}|=\mathcal{K}$. 
Each index $i \in \mathcal{S}$ have a unique computational basis state $\ket{i} \equiv \ket{b_{n-1} \dots b_0}$, 
where the bit string $b_{n-1} \dots b_0$ represents the $n$-bit binary expansion of the integer $i$ for $n$-qubit system. 
The coefficients $c_i$ are random complex numbers with uniform magnitude, $|c_i| = 1/\sqrt{\mathcal{K}}$, 
to ensure the normalization $\braket{\mathcal{K}|\mathcal{K}} = 1$, and differ only by their stochastic phases. 
Thus, $\ket{\mathcal{K}}$ effectively represents the amplitude encoding of the classical array with sparsity $\mathcal{K}$
and elements of uniform magnitude.
Shown in \figref{fig:randomTEE} is the TEE $I^{3}(A,B,C)$, where partitions $A$, $B$, and $C$ are taken as a contiguous subregion of size $n/4$ in a 1D array of qubits (inset). 
As expected from the theory, transition from positive to negative occurs near $\mathcal{K}=2^{n/3}$. 
\begin{figure}[t!]
   \includegraphics[scale=1.0]{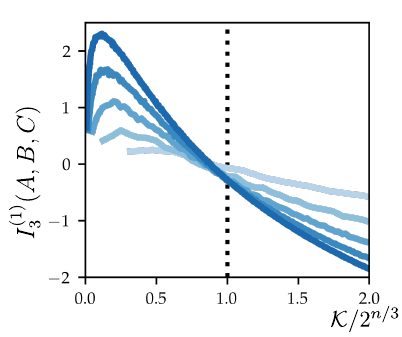}
   \caption{
      {\bf 
      TEE of random state with magnitude unity amplitude}
      TEE $I^{(1)}_3$ of random states with as a function of sparsity $\mathcal{K}$.
      Results are shown for various system sizes $n=8$, $12$, $16$, $20$, and $24$ (light to dark).
      For each $\mathcal{K}$, $\TEE{3}$ is evaluated by averaging over $100$ random states. 
      Error bars are included but may be imperceptible due to their small magnitude.
      \label{fig:randomTEE}
   }
\end{figure}

\section{Derivation of the exact expression for the density matrix of sine MPS}
\label{app:derivrhoinf}
In this appendix, we derive \eqnref{eq:singleqrdmat} of the Methods, which is the exact expression for the 
reduced density matrix of the $q$\textsuperscript{th} qubit in the thermodynamic limit
for a state in which a sine function of wavelength $\lambda$ is encoded on the domain $0\leq x<1$. 

\subsection{Amplitude encoded sine function in MPS form}
As defined in the main article, we define an amplitude encoded single sine function in domain $0\leq x <1$
\begin{align}
   \ket{\psi_{\lambda}^n} \propto\sum_{i}\sin\left(\frac{2\pi}{\lambda} x_i + \phi\right)\ket{\bm{\sigma}(x_i)},
\end{align}
where $\lambda$ is a wavelength and $\phi$ is the phase and $x_i$ are the equally spaced, enumerated, grid points over the domain. 
Here, $\ket{\bm{\sigma}_{x_i}}$ is a basis state corresponds to the $i$\textsuperscript{th}  grid point $x_i$, 
obtained from the binary expansion of $i$. 

As a Matrix Product States \cite{oseledetsConstructiveRepresentationFunctions2013},
$\ket{\psi_{\lambda}^n}$ is written as
\begin{align}
   \ket{\psi_{\lambda}^n} = \sum_{\bm{\sigma}}
   \sum_{\bm{\kappa}} M^{(\sigma_0)}_{\kappa_0}
   M^{(\sigma_1)}_{\kappa_0,\kappa_1}
   \ldots
   M^{(\sigma_{n-1})}_{\kappa_{n-2}}\ket{\bm{\sigma}},
\end{align}
where $M^{\sigma_k}_{i,j}$ denotes a rank three tensor and summation is taken over all the possible
basis-state configurations $\bm{\sigma}$ and bond indices $\bm{\kappa}$.
The tensors $M$ are
\begin{align}
   M_{\kappa_0}^{(0)} (\tau) &= \begin{pmatrix} \sin\left(\frac{k}{2}+\tau\right) & \cos\left(\frac{k}{2}+\tau\right) \end{pmatrix}, 
   M_{\kappa_0}^{(1)} (\tau) = \begin{pmatrix} \cos(\tau) & \sin(\tau) \end{pmatrix} \nonumber\\
   M_{\kappa_{q-1},\kappa_{q}}^{(0)} (\tau) &= 
   \begin{pmatrix}
       \cos\left(\frac{k}{2^{q+1}}+\tau\right) & \sin\left(\frac{k}{2^{q+1}}+\tau\right) \\
      -\sin\left(\frac{k}{2^{q+1}}+\tau\right) & \cos\left(\frac{k}{2^{q+1}}+\tau\right)
   \end{pmatrix},
   M_{\kappa_{q-1},\kappa_{q}}^{(1)}  (\tau)= 
   \begin{pmatrix}
      \cos(\tau) & \sin(\tau) \\
      -\sin(\tau) & \cos(\tau)
   \end{pmatrix}\nonumber\\
   M_{\kappa_{n-2}}^{(0)}  (\tau)&= 
   \begin{pmatrix} 
      \cos\left(\frac{k}{2^n}+\tau\right) \\ \sin\left(\frac{k}{2^n}+\tau\right) 
   \end{pmatrix},
   M_{\kappa_{n-2}}^{(1)}  (\tau)= \begin{pmatrix} \sin(\tau) \\ \cos(\tau) \end{pmatrix}. 
\end{align}
where $k=\frac{2\pi}{\lambda}$ is the wavenumber. 
To encode a sine function with non-zero $\phi$, the tensor of one of the chosen site is set to
$\tau=\phi$, and otherwise set $\tau=0$. For the rest of this section, we set $\tau=0$ unless otherwise stated. 

In the computational basis $\ket{0}$ and $\ket{1}$, the reduced density operator of qubit $q$, 
$\rho_q=\mathrm{Tr}_{\bar{q}} \ket{\Psi_{\lambda}^n}\bra{\Psi_{\lambda}^n}$ is obtained by contracting all 
other qubits $\bar{q} \neq q$ and contracting all tensors over the 
the physical indices (upper indices in paranthesis), except for the index of qubit $q$. 

\subsection{Left environment tensor}
To compute the partial trace, we first look at the exact expression of the left environment tensor $\tcl(q)$, 
over the first $q$ qubits, which corresponds to the $q$ largest length scales. 
In terms of local tensors $M$, $\tcl$ is written as
\begin{align}
   \tcl(q)[ \kappa_{q-1}',\kappa_{q-1}]  = \sum_{\bm{\sigma}_{\bar{q}}} \sum_{\bm{\kappa}_q,\bm{\kappa}_q'}
   M^{(\sigma_{0})}_{\kappa_{0}}\ldots M^{(\sigma_{q-1})}_{\kappa_{q-2,\kappa_{q-1}}}
   (M^{(\sigma_{0})}_{\kappa_{0}})^*\ldots (M^{(\sigma_{q-1})}_{\kappa_{q-2,\kappa_{q-1}}})^*
   \nonumber\\
   \label{eq:tclq}
\end{align}
where $M^*$ denotes the element-wise complex conjugation of $M$ and 
$\bm{\sigma}_{q}  = (\sigma_{0},\sigma_{1},\cdots,\sigma_{q-1} )   $,
$\bm{\kappa}_{q}  = (\kappa_{0},\kappa_{1},\cdots,\kappa_{q-2} )   $, and
$\bm{\kappa}_{q}' = (\kappa_{0}',\kappa_{1}',\cdots,\kappa_{q-2}') $
denote the indices of partial summation .

For $q=1$, $\tcl(q)$ is written as 
\begin{align}
   \tcl(1) = 
   \begin{pmatrix}
      \frac{1}{2}\left( 3 + \cos k \right) & \frac{1}{2} \sin k \\ 
      \frac{1}{2}\sin k & \frac{1}{2}\left( 1-\cos k  \right)
   \end{pmatrix}
\end{align}
where, the row dimension goes over the index $\kappa_{q-1}'$ 
and column dimension goes over the index $\kappa_{q-1}$. 

Now we assume that expression for $\tcl(r)$ at $q=r$ is 
\begin{align}
   \tcl(r) &=
   \begin{pmatrix}
      \frac{1}{2} \left[ 2^{r} + \sum_{m=0}^{2^{r}-1} \cos \left( k \frac{m}{2^{r-1}} \right) \right] &
      \frac{1}{2} \sum_{m=0}^{2^{r}-1} \sin \left(  k\frac{m }{2^{r-1}} \right) \\
      \frac{1}{2} \sum_{m=0}^{2^{r}-1} \sin \left(  k\frac{m }{2^{r-1}} \right) &
      \frac{1}{2} \left[ 2^{r} - \sum_{m=0}^{2^{r}-1} \cos \left( k \frac{m}{2^{r-1}} \right) \right]
   \end{pmatrix}
\end{align}
and this is is true for $r=1$.

Now we consider generating $\tcl(r+1)$ from $\tcl(r)$. From the definition of $\tcl(r+1$ (\eqnref{eq:tclq}), 
we get
\begin{align*}
   \tcl(r+1) &= \sum_{\kappa_{r-1},\kappa_{r-1}'} \tcr(r)
\sum_{\sigma_{r-1}}M^{(\sigma_{r})}_{\kappa_{r-1},\kappa_{r}} 
(M^{(\sigma_{r})}_{\kappa_{r-1}',\kappa_{r}'})^* \\
   &= \sum_{\sigma_{r}}\sum_{\kappa_{r-1},\kappa_{r-1}'} 
(M^{(\sigma_{r})}_{\kappa_{r-1}',\kappa_{r}'})^*
\tcl(r)
M^{(\sigma_{r})}_{\kappa_{r-1},\kappa_{r}} 
\end{align*}
This is nothing but a sum of two matrix products 
\begin{align*}
   \tcl(r+1) &=
\begin{pmatrix}
   \cos\frac{k}{2^{r+1}} & -\sin\frac{k}{2^{r+1}} \\
   \sin\frac{k}{2^{r+1}} &  \cos\frac{k}{2^{r+1}}
\end{pmatrix}
\tcl(r)
\begin{pmatrix}
    \cos\frac{k}{2^{r+1}} & \sin\frac{k}{2^{r+1}} \\
   -\sin\frac{k}{2^{r+1}} & \cos\frac{k}{2^{r+1}}
\end{pmatrix}  
+ \tcl(r) \\
             & =
\begin{pmatrix}
   \tcl(r+1)[0,0] & \tcl(r+1)[0,1] \\
   \tcl(r+1)[1,0] & \tcl(r+1)[1,1]
\end{pmatrix},
\end{align*}
where
\begin{align*}
   \tcl(r+1)[0,0]
   =&
   \frac{1}{2}\left[ 2^{r} + \sum_{m=0}^{2^{r}-1}\cos k\frac{m}{2^{r-1}}\cos k\frac{1}{2^r}  
   - \sum_{m=0}^{2^{r}-1} \sin \left(  k\frac{m }{2^{r-1}} \right)\sin k \frac{1}{2^r} \right]
   +
   \frac{1}{2}\left[ 2^{r} + \sum_{m=0}^{2^{r}-1}\cos k\frac{m}{2^{r-1}}  \right] \\
   =&
   \frac{1}{2}\left[ 2^{r+1} + 
   \frac{1}{2} \sum_{m=0}^{2^{r-1}} \cos k\frac{2m+1}{2^{r}} + \cos k\frac{2m-1}{2^{r}}
   +\cos k\frac{2m}{2^{r}}  \right.\\
   &\left. 
    + 
    \frac{1}{2} \sum_{m=1}^{2^{r}-1} 
   \cos\left(  k\frac{2m+1}{2^{r}} \right) - \cos \left(  k\frac{2m-1}{2^{r}} \right)
   \right]\\
   =&
   \frac{1}{2}\left[ 2^{r+1} + 
   \sum_{m=0}^{2^{r}-1} \cos k\frac{2m+1}{2^{r}} +\cos k\frac{2m}{2^{r}}  \right] 
   = 
   \frac{1}{2}\left[ 2^{r+1} + \sum_{m=0}^{2^{r+1}-1} \cos k\frac{m}{2^{r}} \right]
\end{align*}
and we obtained the desired result. 
Similarly for other elements, we obtain
\begin{align*}
   \tcl(r+1)[0,1] = &\tcl(r+1)[1,0] \\
   =&
   \frac{1}{2}\left[ \sum_{m=0}^{2^{r-1}}\cos k\frac{m}{2^{r-1}}\sin k\frac{1}{2^r}
   + \sum_{m=0}^{2^{r}-1} \sin \left(  k\frac{m }{2^{r-1}} \right)\cos k \frac{1}{2^r} 
   + \sum_{m=0}^{2^{r}-1} \sin \left(  k\frac{m }{2^{r-1}} \right)\right] \\
   =&
   \frac{1}{2}
   \sum_{m=0}^{2^{r+1}-1} \sin \left(  k\frac{m}{2^{r}} \right)
\end{align*}
and
\begin{align*}
   \tcl(r+1)[1,1] 
   =&
   \frac{1}{2}\left[ 2^{r} - \sum_{m=0}^{2^{r}-1}\cos k\frac{m}{2^{r-1}}\cos k\frac{1}{2^r}  
   + \sum_{m=0}^{2^{r}-1} \sin \left(  k\frac{m }{2^{r-1}} \right)\sin k \frac{1}{2^r} \right]
   +
   \frac{1}{2}\left[ 2^{r} - \sum_{m=0}^{2^{r}-1}\cos k\frac{m}{2^{r-1}}  \right] \\
   =&
   \frac{1}{2}\left[ 2^{r+1} 
   -\sum_{m=0}^{2^{r}-1} \cos k\frac{2m-1}{2^{r}} -\cos k\frac{2m}{2^{r}}  \right] 
   = 
   \frac{1}{2}\left[ 2^{r+1} - \sum_{m=0}^{2^{r+1}+1} \cos k\frac{m}{2^{r}} \right]
\end{align*}
as desired. 
\subsection{Right environment}
Similarly, we now consider the case of contracting from the right. Here we obtain the right environment $\tcr(q)$, 
which corresponds to the $q$ finest length scales. In terms of local tensors $M$, $\tcr(q)$ is written as
\begin{align}
   \tcr(q)[ \kappa_{n-q-1},\kappa_{n-q-1}' ]  = \sum_{\bm{\sigma}_{\bar{q}}}\sum_{\bm{\kappa}_{\bar{q}},\bm{\kappa}_{\bar{q}}'}
   M^{(\sigma_{n-q})}_{\kappa_{n-q-1},\kappa_{n-q}}\ldots M^{(\sigma_{n-1})}_{\kappa_{n-2}}
   (M^{(\sigma_{n-q})}_{\kappa_{n-q-1}',\kappa_{n-q}'})^* \ldots (M^{(\sigma_{n-1})}_{\kappa_{n-2}'})^*
   \nonumber\\
   \label{eq:deftrfq}
\end{align}
where 
$\bm{\sigma}_{\bar{q} } = \sigma_{n-q} , \sigma_{n-q+1} ,\cdots,\sigma_{n-1}$,
$\bm{\kappa}_{\bar{q} } = \kappa_{n-q-1} , \kappa_{n-q} ,\cdots,\kappa_{n-2}$, and
$\bm{\kappa}_{\bar{q}'} = \kappa_{n-q-1}', \kappa_{n-q}',\cdots,\kappa_{n-2}'$
denote the indices of partial summation. 

Similarly to deriving $\tcr(q)$, we first observe that 
\begin{align}
   \tcr(1) = 
   \begin{pmatrix}
      \frac{1}{2}\left( 1-\cos k2^{1-n}  \right) & \frac{1}{2} \sin k 2^{1-n} \\
      \frac{1}{2} \sin k 2^{1-n} & \frac{1}{2}\left( 3+\cos k2^{1-n}  \right)
   \end{pmatrix}.
\end{align}
where the row dimension goes over index $\kappa_{n-q-1}$ and the column dimension goes over index $\kappa_{n-q-1}'$. 

Now we assume that expression of $\tcr(r)$ at $q=r$ is 
\begin{align}
   \tcr(r) = 
   \begin{pmatrix}
      \frac{1}{2} \left[ 2^{r} - \sum_{m=0}^{2^{r}-1} \cos \left( k \frac{2 m }{2^{n}} \right) \right] &
      \frac{1}{2} \sum_{m=0}^{2^{r}-1} \sin \left( k \frac{2 m }{2^{n}} \right) \\
      \frac{1}{2} \sum_{m=0}^{2^{r}-1} \sin \left( k \frac{2 m }{2^{n}} \right) &
      \frac{1}{2} \left[ 2^{r} + \sum_{m=0}^{2^{r}-1} \cos \left( k \frac{2 m }{2^{n}} \right) \right]
   \end{pmatrix},
\end{align}
which is correct for $r=1$. 
We now derive $\tcr(r+1)$ from $\tcr(r)$ by multiplying the local tensors
\begin{align*}
   \tcr(r+1) &= \sum_{\kappa_{n-r-1},\kappa_{n-r-1}'}
   \tcr(r) 
   \sum_{\sigma_{n-r-1}}
   M^{(\sigma_{n-r-1})}_{\kappa_{n-r-2},\kappa_{n-r-1}}
   (M^{(\sigma_{n-r-1})}_{\kappa_{n-r-2}',\kappa_{n-r-1}'})^*  \\
   & = 
   \sum_{\kappa_{r-1},\kappa_{r-1}'}\sum_{\sigma_{n-r-1}}
   M^{(\sigma_{n-r-1})}_{\kappa_{n-r-2},\kappa_{n-r-1}}\tcr(r)(M^{(\sigma_{n-r-1})}_{\kappa_{n-r-2}',\kappa_{n-r-1}'})^* 
\end{align*}
as we have done for deriving $\tcr$, we obtain 
\begin{align*}
   \tcr(r+1) &= 
   \begin{pmatrix}
       \cos\frac{k}{2^{n-r}} & \sin\frac{k}{2^{n-r}} \\
      -\sin\frac{k}{2^{n-r}} & \cos\frac{k}{2^{n-r}} \\
   \end{pmatrix}
   \tcr(r)
   \begin{pmatrix}
       \cos\frac{k}{2^{n-r}} & -\sin\frac{k}{2^{n-r}} \\
       \sin\frac{k}{2^{n-r}} &  \cos\frac{k}{2^{n-r}} \\
   \end{pmatrix}
   +
   \tcr(r)
\end{align*}
and we obtain the desired result. 

\subsection{Reduced density matrix of qubit $q$ in the thermodynamic limit}
To obtain the reduced density matrix of qubit $q$, we contract the relevant $\tcl$, $M$ and $\tcr$
\begin{align*}
   \rho_q = \frac{1}{\sqrt{\sum_{m=0}^{2^n-1}\sin^2 k\frac{m}{2^n}}}
   \sum_{\kappa_{q-1},\kappa_{q-1}',\kappa_{q},\kappa_{q}'}
   \tcl(q)_{\kappa_{q-1}',\kappa_{q-1}}
   M^{(\sigma_q)}_{\kappa_{q-1},\kappa_{q}}(\phi)
   \tcr(n-q-1)_{\kappa_{q},\kappa_{q}'} %
   (M^{(\sigma_q)}_{\kappa_{q-1}',\kappa_{q}'}(\phi))^*,
\end{align*}

In the limit of $n\to\infty$, we obtain 
\begin{align}
  \rho_q^{\infty} &= \lim_{n\to\infty} \rho_q \nonumber\\
          &= 
   \begin{pmatrix}
      \frac{2k-\sec(2^{-q-1} k)
      \left(\sin((2+2^{-q-1})k+2\phi) - \sin(2^{-q-1}k + 2\phi)\right)}{4 k - 2\sin(2(k+\phi)) + 2\sin(2\phi)} & 
      \frac{2k\cos(2^{-q-1} k) - \sec(2^{-q-1}k)\left(\sin(2(k+\phi)) - \sin(2 \phi)\right)}
      {4 k - 2\sin(2(k+\phi)) + 2\sin(2\phi)} \\
      \frac{2k\cos(2^{-q-1} k) - \sec(2^{-q-1}k)\left(\sin(2(k+\phi)) - \sin(2 \phi)\right)}
      {4 k - 2\sin(2(k+\phi)) + 2\sin(2\phi)} &
      \frac{2k-\sec(2^{-q-1} k)
      \left(\sin((2-2^{-q-1})k+2\phi) - \sin(2\phi-2^{-q-1}k)\right)}{4 k - 2\sin(2(k+\phi)) + 2\sin(2\phi)}
   \end{pmatrix} 
\end{align}
for arbitrary $\phi$ as presented in the Methods.

\subsection{Threshold for Amplitude-Encoded Band Limited functions}
\label{app:QNSSTmultiplefreq}
Finally, we comment on the case of the amplitude-encoded state contains multiple wavelengths greater than $\lambda_{\min}$.
If we focus on qubit $q$, we show that the wavefunction $\ket{\Psi_{\lambda}^\infty}$ for $\lambda_{\min} < \lambda$ 
can be written as
\begin{align}
   \ket{\Psi_{\lambda}^\infty} = \ket{\Psi_{\lambda}^{q}}\ket{\bm{+}}
   + \epsilon_{\lambda}\ket{\mathrm{res}_{\lambda}^{\infty}}
\end{align}
where $\ket{\mathrm{res}_{\lambda}^{\infty}}$ is the anything that is left from the Riemannian approximation beyond qubit $q$, 
which we show that the contribution is exponentially small $\epsilon_{\lambda} = \mathcal{O}(2^{q_c(\lambda)-q})$ 
for $q > q_c(\lambda_{\min})=\lceil \log_2\pi/\lambda_{\min} \rceil$, and $\ket{\bm{+}}$ is padded with $\ket{+}$ states
beyond qubit $q$.

Let $c(\lambda)$ be the spectrum of the field composed of multiple frequencies. The total state $\ket{\Psi}$ is written as follows
\begin{align}
   \ket{\Psi} = \int_{\lambda_{\min}}^1 c(\lambda) \ket{\psi_{\lambda}^{\infty}}
\end{align}
with the normalization condition 
\begin{align}
   \braket{\Psi|\Psi} = \int_{\lambda_{\min}}^{1}\int_{\lambda_{\min}}^{1}d\lambda d\lambda' 
   c^*(\lambda') c(\lambda)  \braket{\psi_{\lambda'}^{\infty}|\psi_{\lambda}^{\infty}}= 1
\end{align}
where $c^{*}(\lambda')$ is a complex conjugate of $c(\lambda')$.
The magnitude of the the total residual of 
$\ket{\Psi}$, $\ket{\mathrm{Err}}=\int_{\lambda_{\min}}^{1}  d\lambda c(\lambda) \epsilon_\lambda \ket{\mathrm{res}_\lambda^\infty}$, is then bounded by 
\begin{align}
   \lVert \ket{\mathrm{Err}} \rVert= \sqrt{\braket{\mathrm{Err}|\mathrm{Err}}} \leq 
   \mathcal{O}(2^{q_c(\lambda_{\min}) - q}) \int_{\lambda_{\min}}^{1} d\lambda |c(\lambda)|.  
\end{align}
We further bound $L_1 = \int_{\lambda_{\min}}^{1} d\lambda |c(\lambda)|$ with the Cauchy Schwarz inequality
\begin{align}
   L_1 \leq \sqrt{\int_{\lambda_{\min}}^{1} d\lambda |c(\lambda)|^2} \sqrt{\int_{\lambda_{\min}}^{1} d\lambda 1^2}
   =  \sqrt{\int_{\lambda_{\min}}^{1} d\lambda |c(\lambda)|^2} \sqrt{1-\lambda_{\min}}. 
\end{align}
Where $L_2 = \int_{\lambda_{\min}}^{1} d\lambda |c(\lambda)|^2$ is the $\ell_2$-norm of $c(\lambda)$. 
$L_2$ is bounded by the spectral feature of $\ket{\Psi}$. Let $g_{\min}$ the infimum of the spectrum of a positive definite kernel 
$G(\lambda',\lambda)=\braket{\Psi^{\infty}_{\lambda'}|\Psi^{\infty}_{\lambda}}$, 
Rayleigh coefficient of $G(\lambda',\lambda)$ and $c(\lambda)$ is bounded by $0 < g_{\min}$, giving the bound on $L_2$
\begin{align}
   \frac{\int_{\lambda_{\min}}^{1}\int_{\lambda_{\min}}^{1} d\lambda'd\lambda  c^*(\lambda')G(\lambda',\lambda)c(\lambda)}{L_2}
   \geq g_{\min}
   \ \rightarrow \ L_2 \leq \frac{1}{g_{\min}}
\end{align}
Altogether, we obtain the bound
\begin{align}
   \lVert \ket{\mathrm{Err}} \rVert \leq \mathcal{O}(2^{q_c(\lambda_{\min})-q})\sqrt{\frac{1-\lambda_{\min}}{g_{\min}}}
   = \frac{1}{\sqrt{g_{\min}}}\mathcal{O}(2^{q_c(\lambda_{\min})-q}),
\end{align}
which does not modify the threshold qubit, but strongly affected by the underlying spectral structure 
characterized by the kernel $G(\lambda',\lambda)$ of $\ket{\Psi}$.

\section{Error scaling}
\label{app:errorscaling}
In this appendix, we provide a formal derivation for the scaling of the circuit error under a local error model. 
We consider a brickwork ansatz circuit, discussed in the Methods section, of depth $D$ on $n$ qubits. 
We assume an error model in which each single-qubit gate can be faulty with an error of less than $\varepsilon$, 
while the $\mathrm{CNOT}$ gates are negligible.
We derive a general upper bound on the global error, which measures the distinguishability 
of the final state from the ideal one. 

\newtheorem{theorem}{Theorem}
\begin{theorem}[Global Error Bound]
Let $U$ be the ideal unitary for a circuit composed of $D$ layers of single-qubit gates, and let $\tilde{U}$ be the corresponding 
faulty unitary in which each of the $n \times D$ single-qubit gates $R$ has an error bounded by $\| \tilde{R} - R \| \le \varepsilon$. 
The global error, measured in the operator norm, is bounded by:
$$
\| U - \tilde{U} \| \le n D \varepsilon
$$
The scaling is therefore at most $O(n D \varepsilon)$.
\end{theorem}

\begin{proof}
   Let the ideal circuit be a sequence of $D$ unitary layers (including both single-qubit and $\mathrm{CNOT}$ layers), 
$U = U_{D} U_{D-1} \cdots U_1$. 
The faulty circuit is $\tilde{U} = \tilde{U}_{D} \tilde{U}_{D-1} \cdots \tilde{U}_1$, 
in which $\tilde{U}_k = U_k$ if layer $k$ consists of error-free gate implementations.

The difference between the ideal and faulty unitaries can be expressed as a telescoping sum:
\begin{equation}
   U - \tilde{U} = \sum_{D=1}^{D} \left( U_{>D}U_D\tilde{U}_{D>}- U_{>D} \tilde{U}_D\tilde{U}_{D>} \right)
\end{equation}
where $U_{>D}$ is the unitary composed of the layers at depth larger than $D$, and $U_{D>}$ is the unitary composed of the layers at depth smaller than $D$.

By factoring and applying the triangle inequality, we obtain 
\begin{align}
   \| U - \tilde{U} \| &\le \sum_{D=1}^{D} \| U_{>D}(U_D - \tilde{U}_D) \tilde{U}_{D>} \| \nonumber \\
                       &\le \sum_{D=1}^{D} \| U_D - \tilde{U}_D \|,
\end{align}
where we used the norm-conserving properties of unitaries $U$, $V$ acting on $A$ ($\|UAV\|=\|A\|$). 
Since the $\mathrm{CNOT}$ layers are assumed to be perfect, we are left with a sum over the $D$ single-qubit layers, 
which we denote $U_{SQ, D}$ for $D=1, \dots, D$
\begin{equation}
   \| U - \tilde{U} \| \le \sum_{D=1}^{D} \| U_{\mathrm{SQ},D} - \tilde{U}_{\mathrm{SQ},D} \|. 
\end{equation}
For a single layer of $n$ single-qubit gates, $U_{\mathrm{SQ},D} = \bigotimes_{i=0}^{n-1} R_{D,i}$. The faulty layer is 
$\tilde{U}_{\mathrm{SQ},D} = \bigotimes_{i=0}^{n-1} \tilde{R}_{D,i}$. 
The difference is, again, by applying the telescoping sum
\begin{align}
   U_{\mathrm{SQ},D} - \tilde{U}_{\mathrm{SQ},D} 
   &= \sum_{i=0}^{n-1}\tilde{R}_{D,>i}(R_{D,i} - \tilde{R}_{D,i})R_{D,i>}.
\end{align}
Similar to bounding $\| U - \tilde{U} \|$ we obtain 
\begin{align}
   \| U_{SQ,D} - \tilde{U}_{SQ,D} \| &\le \sum_{i=0}^{n-1} \| R_{D,i} - \tilde{R}_{D,i} \| \le \sum_{i=0}^{n-1} \epsilon = n\epsilon. 
\end{align}
Substituting this layer error back into the sum over all $D$ single-qubit layers gives the final bound:
\begin{equation}
   \| U - \tilde{U} \| \le \sum_{D=1}^{D} (n\epsilon) = n D \epsilon
\end{equation}
Since $nD$ is equivalent to the number of parameters in BAC, we obtain the error bound provided in the main text. 
\end{proof}
The bound trivially generalizes for inclusion of two-qubit gate implementation error $\epsilon_2$, 
which adds the extra $n$ and $\epsilon$ dependence ($n D \epsilon + \frac{1}{2}nD \epsilon_2$)
but does not change the overall scaling, and other ansatz format (sum of the errors over all gates).
Since for given $\lambda_c$, the maximum number of qubits required is limited by the quantum Nyquist-Shannon sampling theorem $n\sim -\log_2(\lambda_c)$, we obtain the scaling discussed in the main text. 

\section{Scaling of minimum parameters required for encoding Weierstrass function}
\label{app:weierstrassencoding}

In this appendix, we analyze the scaling of the minimum number of parameters required to variationally encode the Weierstrass function $W_a(x)$ with an infidelity below a given threshold. 
We employ a brickwork ansatz circuit (BAC) where the single-qubit rotations are restricted to the $y$-axis of the Bloch sphere 
($R_{d,i}=\exp(\frac{\theta_{d,i}}{2}\sigma^{y}_i)$). This choice ensures that the amplitudes of the resulting quantum state remain real-valued, which is consistent with the real-valued nature of the target Weierstrass function $f(x)=W_a(x)$ shown in \figref{fig:scaling}a.

\begin{figure*}[t!]
   \includegraphics[scale=1.0]{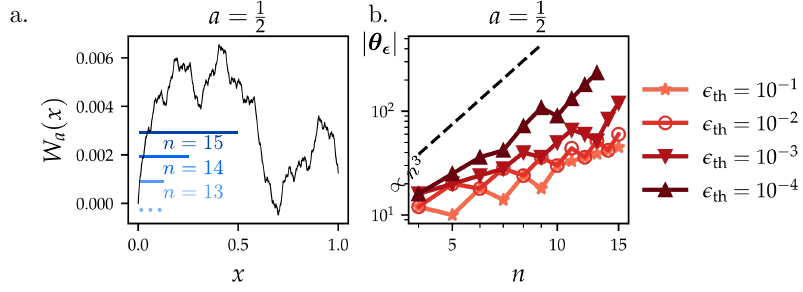}
   \caption{
   \textbf{Parameter scaling for function encoding using brickwork ansatz circuits.}
   \textbf{a.}~Target function: The Weierstrass function $f(x)=W_a(x)$ with $a=1/2$, 
   defined on $0\leq x<1$ and discretized with grid spacing $\Delta x = 2^{-16}$. 
   The horizontal lines in different shades of blue depict the domain for an $n$-qubit state.
   \textbf{b.}~Minimum number of parameters,
   required to encode the function $f(x)$ with an infidelity below a threshold $\epsilon_{\mathrm{th}}$.
   For each $n$ (up to $n=15$), the first $2^n$ points of the discretized function from (\textbf{a}) 
   are variationally encoded with the optimal set of parameters provided in \cite{hashizumedataset2026}.
   A black dashed line shows the cubic scaling, and serves as a guide.
   \label{fig:scaling}}
\end{figure*}

The infidelity $\varepsilon(\bm{\theta})$ between the target state $\ket{W_a}$ and the variational state $U(\bm{\theta})\ket{0}$ is defined as
\begin{equation}
   \varepsilon(\bm{\theta}) = 1-|\bra{W_a}U(\bm{\theta})\ket{0}|^2.
\end{equation}
For each $n \in \{2, \dots, 15\}$, we take the first $2^n$ points of the discretized $W_a(x)$ as the target function values, as shown in \figref{fig:scaling}a. The resulting minimum number of parameters $|\bm{\theta}_{\epsilon}|$ required to achieve $\varepsilon \leq \epsilon_{\mathrm{th}}$ for various $n$ and $\epsilon_{\mathrm{th}}$ is plotted in \figref{fig:scaling}b.

For high-precision encoding with $\epsilon_{\mathrm{th}} \leq 10^{-3}$, we find that the scaling approaches the expected $O(n^3)$ behavior, which is consistent with the complexity of the fractal-like features of the Weierstrass function. In contrast, for a larger threshold of $\epsilon_{\mathrm{th}} \geq 0.01$, the growth of $|\bm{\theta}_{\epsilon}|$ is significantly slower, suggesting that the practically required circuit depth may be much lower if high precision is not required. These findings provide further evidence that VQAs offer a quantum advantage in regimes dominated by fine-scale features, where classical computational costs become intractable.

\end{document}